\documentclass[runningheads,12pt]{llncs}
\usepackage[pdftex]{graphicx}
\usepackage{amsmath,amssymb,amsfonts}
\usepackage[textsize=small,textwidth=1.3in]{todonotes}
\usepackage[english]{babel}

\usepackage{hyperref}

\usepackage{authblk}

\usepackage{multirow}

\usepackage{tikz}
\usepackage{color, colortbl}
\usepackage{multirow}
\usepackage{easybmat}

\usepackage{amsthm}

\newcommand{\bc}{{\bf c}}

\newcommand{\bx}{{\bf x}}

\newcommand{\rk}{\rm{rank}}

\begin{document}
\title{Entanglement-assisted Quantum Codes from Cyclic Codes}
\titlerunning{QUENTA codes from Cyclic Codes}
%
\author{Francisco Revson F. Pereira\inst{1,2}}
\authorrunning{Francisco Revson F. Pereira}
%
\institute{Department of Mathematics and Computing Science, Eindhoven University of Technology, Eindhoven, The Netherlands. \and
	   Department of Electrical Engineering, Federal University of Campina Grande, Campina Grande, Para\'iba, Brazil.\\
\email{revson.ee@gmail.com}}
%
\maketitle              
\begin{abstract}

Entanglement-assisted quantum (QUENTA) codes are a subclass of quantum error
correcting codes which use entanglement as a resource. These codes can provide
error correction capability higher than the codes derived from the traditional stabilizer formalism.
In this paper, it is shown a general method to construct QUENTA codes
from cyclic codes. Afterwards, the method is applied to Reed-Solomon codes, BCH codes,
and general cyclic codes. We use the Euclidean and Hermitian construction of QUENTA codes.
Two families of QUENTA codes are maximal distance separable (MDS), and one is almost MDS or
almost near MDS. The comparison of the codes in this paper is mostly based on the quantum
Singleton bound.

\newline

\end{abstract}
%
%
%


\section{Introduction}
\label{sec:Introduction}
\noindent

Practical implementations of most quantum communication schemes and quantum computers will only be
possible if such systems incorporate quantum error correcting codes to them. Quantum error correcting codes
restore quantum states from the action of a noisy quantum channel. One of the most known and used methods to
create quantum codes from classical block codes is the CSS method \cite{NielsenChuang:Book}. Unfortunately, it requires (Euclidean or
Hermitian) duality containing to one of the classical codes used. One way to overcome this constraint is via 
entanglement. It is also possible to show that entanglement also improves the error-correction capability of quantum codes.
These codes are called Entanglement-Assisted Quantum (QUENTA) codes, also denoted by EAQEC codes in the literature. The first
proposals of QUENTA codes were done by Bowen \cite{Bowen:2002} and Fattal, \emph{et al.} \cite{Fattal:2004}. In the following,
Brun \emph{et al.} \cite{Brun:2006} have developed an entanglement-assisted stabilizer formalism for these codes, which were recently generalized
by Galindo, \emph{et al.} \cite{Galindo:2019}.

This formalism has created a method to construct QUENTA codes from classical block codes, which has lead to
the construction of several families of QUENTA codes \cite{Wilde:2008,Fan:2016,Chen:2017,Lu:2018,Chen:2018,Chen:2017,Lu:2017,Guenda:2018,Liu:2019}.
The majority of them utilized constacyclic codes \cite{Fan:2016,Chen:2018,Lu:2018} or
negacyclic codes \cite{Chen:2017,Lu:2018} as the classical counterpart.
However, only a few of them have used cyclic codes and described the parameters of the quantum code constructed via
the defining set of cyclic code. This can lead to a straightforward relation between the parameters of
the classical and quantum codes, and a method to create MDS QUENTA code. Li \emph{et al.} used BCH codes to
construct QUENTA codes via decomposing the defining set of the BCH code used \cite{Li:2011}. Lu and Li constructed
QUENTA codes from primitive quaternary BCH codes \cite{LuLi:2014}. Recently, Lu et al. \cite{Lu:2018}, using not cyclic but
constacyclic MDS codes as the classical counterpart, proposed four families of MDS QUENTA codes.

The main goal of this paper is to describe any cyclic code, such as Reed-Solomon and BCH codes, under the same framework via defining
set description, and to show, using two classical codes from one of these families, how to construct QUENTA codes from them.
We have used the Euclidean and Hermitian methods to construct QUENTA codes. And as it will be shown, QUENTA codes from Reed-Solomon codes
are MDS codes and the ones from BCH codes are new in two senses. The first one is that there is no work in the literature with the same parameters.
The second is that we are using two BCH codes to derive the QUENTA code, this gives more liberty in the choice of parameters.
Two more families of QUENTA codes are constructed using the Hermitian construction. One of these families can generate codes which
are almost MDS or almost near MDS; i.e., the Singleton defect for these codes is equal to one or two units. The last family created is
maximal entangled and have length proportional to a high power of the cardinality of the field, which make them suitable to achieve
the hashing bound \cite{Li:2014}. Lastly, we would like to highlight that the description made in this paper gives a more direct
relation between cyclic codes and the entanglement-assisted quantum codes constructed from them. Furthermore, such relation can be
extended to constacyclic and negacyclic codes with a few adjustments.

The paper is organized as follows. In Section~\ref{sec:Preliminaries}, we review Reed-Solomon and BCH codes and
describe their parameters via defining set. Additionally, it is shown construction methods of QUENTA codes from
classical codes. Using these methods to cyclic classical, new QUENTA codes are
constructed in Section~\ref{sec:NewConstructions}. In Section~\ref{sec:codComp}, a comparison of these codes with
the quantum Singleton bound is covered. In particular,
it is shown families of MDS and almost MDS QUENTA codes.
We also create a family of QUENTA codes which can be applied to achieve the hashing bound \cite{Li:2014}.
Lastly, the conclusion is carried out in Section~\ref{sec:Conclusion}.

\emph{Notation.} Throughout this paper, $p$ denotes a prime number and $q\neq 2$ is a power of $p$.
Let $\mathbb{F}_q$ be the finite field with $q$ elements. A linear code $C$ with parameters $[n,k,d]_q$
is a $k$-dimensional subspace of $\mathbb{F}_q^n$ with minimum distance $d$. For cyclic codes,
$Z(C)$ denotes the defining set and $g(x)$ is the generator polynomial.
Lastly, an $[[n,k,d;c]]_q$ quantum code is a $q^k$-dimensional subspace
of $\mathbb{C}^{q^n}$ with minimum distance $d$ that utilizes $c$ pre-shared
entangled pairs.


\section{Preliminaries}
\label{sec:Preliminaries}
\noindent
In this section, we review some ideas related to linear complementary dual (LCD) codes, cyclic codes,
and entanglement-assisted quantum (QUENTA) codes. As it will be shown, LCD codes give an important source
to construct QUENTA codes with interesting properties, such as maximal distance separability and maximal
entanglement (see Section~\ref{sec:NewConstructions}). But before giving a description of LCD codes,
we need to define the Euclidean and Hermitian dual of a linear code.


\begin{definition}
%
    Let $C$ be a linear code over $\mathbb{F}_q$ with length $n$. The (Euclidean) dual of $C$ is defined as

    \begin{equation}
    C^\perp = \{ \ \bx \in \mathbb{F}_q^n \ | \ \bx\cdot \bc =0 \mbox{ for all } \bc \in C \}.
    \end{equation}
    If the finite field has cardinality equal to $q^2$, an even power of a prime,
    then we can define the Hermitian dual of $C$.
    This dual code is defined by

    \begin{equation}
	C^{\perp_h} = \{ \ \bx \in \mathbb{F}_{q^2}^n \ | \ \bx\cdot \bc^q =0 \mbox{ for all } \bc \in C \ \},
    \end{equation}where $\bc^q = (c_1^q, \ldots, c_n^q)$ for $\bc \in\mathbb{F}_{q^2}^n$.

%
\end{definition}

These types of dual codes can be used to derive quantum codes from the stabilizer formalism \cite{NielsenChuang:Book}.
The requirement in this formalism is to the classical code to be self-dual; i.e.,
$C\subseteq C^\perp$ or $C\subseteq C^{\perp_H}$. However, there is
a different relationship between a code and its (Euclidean or Hermitian) dual that it can be interesting in the construction
of QUENTA codes. This relation is complementary duality and is defined in the following.

\begin{definition}
    The hull of a linear code $C$ is given by $hull (C) = C^\perp \cap C$. The code is called \rm{linear complementary
    dual} (LCD) code if the hull is trivial; i.e, $hull(C) = \{\bf{0}\}$. Similarly, it is defined
    $hull_H (C) = C^{\perp_h} \cap C$ and the idea of Hermitian LCD code.
\end{definition}

Now, we can define cyclic codes and some properties that can be used to extract the parameters of the quantum code
constructed from them.

%
%

\subsection{Cyclic codes}

A linear code $C$ with parameters $[n,k,d]_q$ is called cyclic if for any codeword
$(c_0, c_1, \ldots, c_{n-1})\in C$ implies
$(c_{n-1}, c_0, c_1, \ldots, c_{n-2})\in C$. Defining a map from $\mathbb{F}_q^n$ to
$\mathbb{F}_q[x]/(x^n-1)$, which takes $\bc = (c_0, c_1, \ldots, c_{n-1})\in \mathbb{F}_q^n$ to
$c(x) = c_0 + c_1 x + \cdots + c_{n-1}x^{n-1}\in \mathbb{F}_q[x]/(x^n-1)$, we can see that a linear code
$C$ is cyclic if and only if it corresponds to an ideal of the ring $\mathbb{F}_q[x]/(x^n-1)$.
Since that any ideal in $\mathbb{F}_q[x]/(x^n-1)$ is principal, then any cyclic code $C$ is generated
by a polynomial $g(x)|(x^n-1)$, which it is called generator polynomial. This polynomial is monic
and has the smallest degree among all the generators of $C$.

A characterization of the parameters of a cyclic code can be given from the generator polynomial and
its defining set. For the description of this set, consider the following: Let
$m = ord_n(q)$, $\alpha$ be a generator of the multiplicative group $\mathbb{F}_{q^m}^*$, and
assume $\beta = \alpha^{\frac{q^m - 1}{n}}$; i.e., $\beta$ is a primitive $n$-th root of unity. Then
the defining set of $C$, which is denoted by $Z(C)$, is defined as
$Z(C) = \{i\in\mathbb{Z}_n\colon c(\beta^i) = 0\text{ for all }c(x)\in C\}$.

BCH and Reed-Solomon codes are particular cases of cyclic codes, where the generator polynomial has some
additional properties. See Definitions~\ref{Definition:BCH}~and~\ref{Definition:RS}.

\begin{definition}
            Let $b\geq 0$, $\delta \geq 1$, and $\alpha\in\mathbb{F}_{q^m}$, where $m = ord_n(q)$. A cyclic code $C$
            of length $n$ over $\mathbb{F}_q$ is a BCH code with designed
            distance $\delta$ if

            \begin{equation*}
                g(x) = \text{lcm}\{m_{b}(x), m_{b+1}(x), \ldots, m_{b+\delta -2}(x)\}
            \end{equation*}where $m_{i}(x)$ is the minimal polynomial of $\alpha^i$ over
            $\mathbb{F}_q$.If $n = q^m - 1$ then the BCH code is called primitive, and if $b = 1$ it is
            called narrow sense.
            \label{Definition:BCH}
\end{definition}

Before relating the parameters of an BCH code with the defining set, we need to introduce the idea of cyclotomic
coset. It comes from the observation that the minimal polynomial $m_i(x)$ of $\alpha^i$ can be the minimal polynomial
of other powers of $\alpha$. The reason for this is that $\alpha$ belongs to an extension of $\mathbb{F}_{q}$ while the
polynomial $m_i(x)\in\mathbb{F}_q[x]$. The set of all zeros of $m_i(x)$ in the field $\mathbb{F}_{q^m}$ is given by
the cyclotomic coset of $i$. Thus, the defining set of a BCH code $C$ is the union of the cyclotomic cosets of 
$b, b+1, \ldots, b+\delta - 2$. The following definition describes this set.

\begin{definition}
The $q$-ary cyclotomic coset $\mod n$ containing an element $i$ is defined by

\begin{equation}
\mathbb{C}_i =\{i,iq,iq^2,iq^3, \ldots,iq^{m_i-1}\},
\end{equation}where $m_i$ is the smallest positive integer such that $iq^{m_i} \equiv i \mod n$.
\end{definition}

For the parameters of a BCH code, it is shown that the dimension is equal to $n-|Z(C)|$ and the minimal distance
of $C$ is at least $\delta$ \cite{Ruud:Book}. Thus, we can see that important properties of an BCH codes can be
obtained from the defining set. The same characterization happens with Euclidean or Hermitian dual
cyclic code. Propositions~\ref{Cyclic_Euclidean_Dual}~and~\ref{Cyclic_Hermitian_Dual}
focus on this.

\begin{proposition}\cite[Proposition 4.3.8]{Ruud:Book}
            Let $C$ be a linear code of length $n$ and defining set $Z(C)$. Then the defining
            set of $C^\perp$ is given by

            \begin{equation*}
                        Z(C^\perp) = \mathbb{Z}_n\setminus\{-i| i\in Z(C)\}
            \end{equation*}For BCH codes, the generator polynomial is given by the lcm of
            the minimal polynomials over $\mathbb{F}_q$ of the elements $\alpha^j$ such that $j\in Z(C^\perp)$.
            \label{Cyclic_Euclidean_Dual}
\end{proposition}

\begin{proposition}
            Let $C$ be a cyclic code over $\mathbb{F}_{q^2}$ with defining set $Z(C)$.
            Then

            \begin{equation*}
                Z(C^{\perp_h}) = \mathbb{Z}_n\setminus\{-i| i\in qZ(C)\}.
            \end{equation*}
            \label{Cyclic_Hermitian_Dual}
\end{proposition}

\begin{proof}
            Let $\bc\in\mathbb{F}_{q^2}^n$ be a codeword of $C$. Expressing $\bc^q$
            as a polynomial we have that $c^{(q)}(x) = c_0^q + c_1^{q}x + \cdots + c_{n-1}^q x^{n-1}$. So,
            $i\in\mathbb{Z}_n$ belongs to $Z(C^q)$ if and only if

            \begin{eqnarray*}
                c^{(q)}(\alpha^i) = 0 &\iff& c_0^q + c_1^{q}\alpha^i + \cdots + c_{n-1}^q \alpha^{i(n-1)} = 0\\
                                    &\iff& (c_0^q + c_1^{q}\alpha^i + \cdots + c_{n-1}^q \alpha^{i(n-1)})^q = 0\\
                                    &\iff& c_0 + c_1\alpha^{iq} + \cdots + c_{n-1}\alpha^{iq(n-1)} = 0\\
                                    &\iff& iq\in Z(C).
            \end{eqnarray*}This shows that $Z(C^q) = qZ(C)$. Since $C^{\perp_h} = (C^q)^\perp$, we have from
            Proposition~\ref{Cyclic_Euclidean_Dual} that $Z(C^{\perp_h}) = \mathbb{Z}_n\setminus\{-i| i\in qZ(C)\}$.
\end{proof}

The other class of cyclic codes used in this paper, Reed-Solomon codes, can be viewed as a subclass of BCH codes. Thus,
a similar characterization in terms of defining set can be given, see Definition~\ref{Definition:RS} and
Corollary~\ref{RS_Dual}. One property of such codes that make them important is that they are maximal distance separable
(MDS) codes; i.e., fixing the length and the dimension, they have the maximal minimal distance possible. As shown in
Section~\ref{sec:NewConstructions}, using such codes to construct QUENTA codes will result in MDS quantum codes.

\begin{definition}
            Let $b\geq 0$, $n = q-1$, and $1\leq k\leq n$. A cyclic code $RS_k(n,b)$
            of length $n$ over $\mathbb{F}_q$ is a \emph{Reed-Solomon code} with minimal
            distance $n-k+1$ if

            \begin{equation*}
                        g(x) = (x-\alpha^b)(x-\alpha^{b+1})\cdot\cdots\cdot(x-\alpha^{b+n-k-1}),
            \end{equation*}where $\alpha$ is a primitive element of $\mathbb{F}_q$.
            \label{Definition:RS}
\end{definition}

A particular application of Proposition~\ref{Cyclic_Euclidean_Dual} to Reed-Solomon codes is described
in Corollary~\ref{RS_Dual}, where the parameters and defining set of an Euclidean dual of a Reed-Solomon
is derived.

\begin{corollary}
            Let $RS_k(n,b)$ be a Reed-Solomon code. Then its Euclidean dual can be described as

            \begin{equation*}
                        RS_k(n,b)^\perp = RS_{n-k}(n,n-b+1)
            \end{equation*}In particular, the defining the of $RS_k(n,b)^\perp$ is given by
            $Z(RS_k(n,b)^\perp) = \{n-b+1, n-b+2, \ldots, n-b+k\}$.
            \label{RS_Dual}
\end{corollary}

As it will be in the next subsection, the amount of entanglement in a QUENTA code is computed from
the dimension of the intersection between two codes. So, the last proposition of this subsection addresses
such subject.

\begin{proposition}\cite[Exercise 239, Chapter 4]{Huffman:2003}
	    Let $C_1$ and $C_2$ be cyclic codes with defining set $Z(C_1)$ and $Z(C_2)$, respectively.
	    Then the defining set of $C_1\cap C_2$ is given by $Z(C_1)\cup Z(C_2)$. In particular,
	    $\dim(C_1 \cap C_2) = n - |Z(C_1)\cup Z(C_2)|$.
	    \label{Intersection_Cyclic}
\end{proposition}

\subsection{Entanglement-assisted quantum codes}

\begin{definition}
  A quantum code $\mathcal{Q}$ is called an $[[n,k,d;c]]_q$ entanglement-assisted quantum (QUENTA) code
  if it encodes $k$ logical qudits into $n$ physical qudits using $c$ copies of maximally entangled states
  and can correct $\lfloor(d-1)/2\rfloor$ quantum errors. A QUENTA code is said to have maximal
  entanglement when $c = n-k$.
\end{definition}

Formulating a stabilizer paradigm for QUENTA codes gives a way to use classical codes to construct this
quantum codes \cite{Brun:2014}. In particular, we have the next two procedures by
Galindo, \emph{et al.} \cite{Galindo:2019}.

\begin{proposition}\cite[Theorem 4]{Galindo:2019}
Let $C_1$ and $C_2$ be two linear codes over $\mathbb{F}_q$ with parameters $[n,k_1,d_1]_q$ and
$[n,k_2,d_2]_q$ and parity check matrices $H_1$ and $H_2$, respectively. Then there is a QUENTA code with parameters
$[[n,k_1+k_2-n+c, d; c]]_q$, where $d = \min\{d_H(C_1\setminus(C_1\cap C_2^\perp)), d_H(C_2\setminus(C_1^\perp\cap C_2))\}$,
with $d_H$ as the minimum Hamming weight of the vectors in the set, and

\begin{equation}
  c = {\rk} (H_1 H_2^T) = \dim C_1^\perp - \dim (C_1^\perp\cap C_2)
\end{equation}is the number of required maximally entangled states.
\label{Prep:WildeEuclid}
\end{proposition}


\begin{proposition}\cite[Proposition 3 and Corollary 1]{Galindo:2019}
Let $C$ be a linear codes over $\mathbb{F}_{q^2}$ with parameters $[n,k,d]_q$,
$H$ be a parity check matrix for $C$, and $H^*$ be the $q$-th power of the transpose matrix of $H$.
Then there is a QUENTA code with parameters
$[[n,2k-n+c, d'; c]]_q$, where $d' = d_H(C\setminus(C\cap C^{\perp_h}))$,
with $d_H$ as the minimum Hamming weight of the vectors in the set, and

\begin{equation}
  c = {\rk} (H H^*) = \dim C^{\perp_h} - \dim (C^{\perp_h}\cap C)
\end{equation}is the number of required maximally entangled states.
\label{Prep:WildeHerm}
\end{proposition}

A measurement of goodness for a QUENTA code is the quantum Singleton bound (QSB).
Let $[[n,k,d;c]]_q$ be a QUENTA code, then the QSB is given by

\begin{equation}\label{QSB}
            d \leq  \Big{\lfloor}\frac{n-k+c}{2}\Big{\rfloor} + 1.
\end{equation}The difference between the QSB and $d$ is called \emph{quantum Singleton defect}. When the quantum Singleton
defect is equal to zero (resp. one) the code is called maximum distance separable quantum code (resp. almost maximum
distance separable quantum code) and it is denoted MDS quantum code (resp. almost MDS quantum code).

%
%

\section{New Entanglement-Assisted Quantum Error Correcting Cyclic Codes}
\label{sec:NewConstructions}
In this section is shown the construction of QUENTA codes from the cyclic codes.
We are going to make use of Euclidean and Hermitian constructions, which will give codes with different
parameters when compared over the same field.

\subsection{Euclidean Construction}

A straightforward application of cyclic codes to the Proposition~\ref{Prep:WildeEuclid} via
defining set description can produce some interesting
results. See Theorem~\ref{newConstruction:theorem1} and Corollary~\ref{corollary2}.

\begin{theorem}
    Let $C_1$ and $C_2$ be two cyclic codes with parameters $[n,k_1,d_1]_q$ and $[n,k_2,d_2]_q$,
    respectively. Then there is an QUENTA code with parameters
    $[[n,k_1-|Z(C_1^\perp)\cap Z(C_2)|, \min\{d_1,d_2\};n-k_2-|Z(C_1^\perp)\cap Z(C_2)|]]_q$.
    \label{newConstruction:theorem1}
\end{theorem}

\begin{proof}
            From Proposition~\ref{Intersection_Cyclic} we have that
            $\dim(C_1^\perp \cap C_2) = n - |Z(C_1^\perp)\cup Z(C_2)| =
            n - |Z(C_2)| - |Z(C_1^\perp)| + |Z(C_1^\perp)\cap Z(C_2)| = k_2 - k_1 + |Z(C_1^\perp)\cap Z(C_2)|$.
            So, the amount of entanglement used in an QUENTA code constructed from these two cyclic codes can
            be computed from Proposition~\ref{Prep:WildeEuclid}, which is
            $c = n-k_2 - |Z(C_1^\perp)\cap Z(C_2)|$. Substituting this value of $c$ in the parameters of the
            QUENTA code in Proposition~\ref{Prep:WildeEuclid}, we obtain an
            $[[n,k_1-|Z(C_1^\perp)\cap Z(C_2)|, \min\{d_1,d_2\};n-k_2-|Z(C_1^\perp)\cap Z(C_2)|]]_q$
            QUENTA code.
\end{proof}

\begin{corollary}
    Let $C$ be a LCD cyclic code with parameters $[n,k,d]_q$. Then there is a maximal entanglement QUENTA code
    with parameters $[[n,k,d;n-k]]_q$. In particular, if $C$ is MDS, so it is the QUENTA code derived from it.
    \label{corollary2}
\end{corollary}

\begin{proof}
            Let $C_1 = C_2 = C$ in Theorem~\ref{newConstruction:theorem1}.
            Since $C$ is LCD, then $|Z(C_1^\perp)\cap Z(C_2)|=0$.
            From Theorem~\ref{newConstruction:theorem1}
            we have that there is an QUENTA code with parameters $[[n,k, d; n-k]]_q$.
\end{proof}

\begin{theorem}
            Let $C_1 = RS_{k_1}(n,b_1)$ and $C_2 = RS_{k_2}(n,b_2)$ be two Reed-Solomon codes
            over $\mathbb{F}_q$ with $0\leq b_1\leq k_1$, $b_2\geq 0$, and $b_1+b_2\leq k_2+1$. Then we have
            two possible cases:
            \begin{enumerate}
              \item For $k_1-b_1\geq b_2$, there is an QUENTA code with parameters $$[[n,b_1+b_2 - 1, n-\min\{k_1,k_2\}+1; n+b_1+b_2-k_1-k_2-1]]_q;$$
              \item For $k_1-b_1< b_2$, there is an QUENTA code with parameters $$[[n,k_1, n-\min\{k_1,k_2\}+1; n-k_2]]_q.$$
            \end{enumerate}
            \label{Theorem:RS_Euclid}
\end{theorem}

\begin{proof}
            From Corollary~\ref{RS_Dual}, we have that $Z(C_1^\perp) = \{n-b_1+1, n-b_1+2, \ldots, n-b_1+k_1\}$.
            First of all, notice that the restriction $b_1+b_2\leq k_2+1$ implies that the first element in
            the defining set of $Z(C_1^\perp)$ comes after the last element in $Z(C_2)$. Since that
            $0\leq b_1\leq k_1$, we have that $n-b_1+k_1\geq n$, which implies in a defining set for
            $C_1^\perp$ equals to $Z(C_1^\perp) = \{n-b_1+1, n-b+2, \ldots,n-1, 0, 1, \ldots, k_1-b_1\}$. Thus,
            $Z(C_1^\perp)$ intersects with $Z(C_2)$ if and only if $k_1-b_1\geq b_2$. In the case that it does,
            the intersection is equals to $Z(C_1^\perp)\cap Z(C_2) = k_1-(b_1+b_2)+1$. The missing claims are
            obtained using these results in Theorem~\ref{newConstruction:theorem1}.
\end{proof}

\begin{corollary}
            Let $C = RS_{k}(n,b)$ be a Reed-Solomon codes over $\mathbb{F}_q$
            with $0< b\leq (k+1)/2$ and $0<k<n\leq q$. Then there is an MDS QUENTA code with parameters
            $[[n,2b - 1, n-k+1; n+2b-2k-1]]_q$. In particular, for $b = (k+1)/2$, there
            is a maximal entanglement MDS QUENTA code.
            \label{MDS_QUENTA_RS}
\end{corollary}

\begin{proof}
            Let $C_1 = C_2 = RS_{k}(n,b)$ in Theorem~\ref{Theorem:RS_Euclid}. Assuming $0\leq b<(k+1)/2$,
            we have that the classical codes attain the first case of Theorem~\ref{Theorem:RS_Euclid}; and for
            $b = (k+1)/2$, we are in the second case of Theorem~\ref{Theorem:RS_Euclid}. Thus,
            substituting the values of $k_1,k_2$ and $b_1,b_2$ by $k$ and $b$, respectively, the result follows.
\end{proof}

In a similar way, we can use BCH codes to construct QUENTA codes. The gain in using BCH codes is that the length
of the code is not bounded by the cardinality of the finite field used. However, creating classical or quantum codes from
BCH codes which are MDS is a difficult task. Our proposal to have BCH codes as the classical counterpart in this
paper is to show how to use two BCH codes to construct
QUENTA codes. In addition, it is also constructed maximal entanglement QUENTA codes. In order to do this, we show
suitable properties concerning some cyclotomic cosets for $n = q^2 - 1$.

\begin{lemma}
            Let $n = q^2 - 1$ with $q>2$. Then the $q$-ary coset $\mathbb{C}_0$ has one element, and
            $\mathbb{C}_i = \{i, iq\}$ for any $1 \leq i \leq q-1$.
            \label{lemma:coset}
\end{lemma}

\begin{proof}
            The first claim is trivial. For the second one, notice $iq^2 \equiv i \mod (q^2 - 1)$. Thus, the only elements in
            $\mathbb{C}_i$ are $i$ and $iq$, for $1 \leq i \leq q-1$.
\end{proof}

From Lemma~\ref{lemma:coset}, we can construct QUENTA codes with length $n = q^2 - 1$. See Theorem~\ref{newConstruction:EuclideanBCH}.

\begin{theorem}
            Let $n = q^2 - 1$ with $q>2$. Assume $a, b$ are integer such that $0\leq a \leq q-1$ and
            $1\leq b\leq q$.
            Then there is an QUENTA code with parameters
            \begin{itemize}
              \item $[[n, 2(q-b)-1, b+1; 2(q-a-1)]]_q$, if $a\geq q-b$ and $b<q$;
              \item $[[n, 2a+1, b+1; 2b-\lfloor \frac{b}{q}\rfloor]]_q$, if $a< q-b$.
            \end{itemize}
            \label{newConstruction:EuclideanBCH}
\end{theorem}

\begin{proof}
            First of all, assume that $C_1^\perp$ has defining set given by $Z(C_1^\perp) = \cup_{i=0}^a \mathbb{C}_i$ and
            the defining set of $C_2$ is equal to $Z(C_2) = \cup_{i = 1}^b \mathbb{C}_{q-i}$.
            From Lemma~\ref{lemma:coset} we have that $|Z(C_1^\perp)| = 2a + 1$ and $|Z(C_2)| = 2b - \lfloor \frac{b}{q}\rfloor$.
            Thus, the dimensions of $C_1$ and $C_2$ are equal to $k_1 = |Z(C_1^\perp)| = 2a + 1$ and
            $k_2 = n - |Z(C_2)| = n-2b+\lfloor \frac{b}{q}\rfloor$, respectively.
            To compute $|Z(C_1^\perp)\cap Z(C_2)|$, we have to consider two cases. If $a\geq q-b$, then we have that
            $Z(C_1^\perp)\cap Z(C_2) = \cup_{i = q-b}^a\mathbb{C}_{i}$, which has cardinality given by
            $|Z(C_1^\perp)\cap Z(C_2)| = 2(a-(q-b)+1) - \lfloor\frac{b}{q}\rfloor$, because $|\mathbb{C}_0| = 1$. On the other hand,
            if $a< q-b$, then $|Z(C_1^\perp)\cap Z(C_2)| = 0$. Lastly, since that $a,b\leq q$, $Z(C_1^\perp) = \cup_{i=0}^a \mathbb{C}_i$, 
            and $n = q^2 - 1$ with $q>2$, we can see
            that $d_1 > d_2 = b+1$. Now, applying these results to Theorem~\ref{newConstruction:theorem1}, we have that
            there is a QUENTA code with parameters $[[n, 2(q-b) - 1 + \lfloor \frac{b}{q}\rfloor, b+1; 2(q-a-1)]]_q$, if $a\geq q-b$, or
            a QUENTA code with parameters $[[n, 2a+1, b+1; 2b-\lfloor \frac{b}{q}\rfloor]]_q$.
\end{proof}

\subsection{Hermitian Construction}

In the same way as before, it possible to use cyclic codes to construct
QUENTA codes from the Hermitian construction method of Proposition~\ref{Prep:WildeHerm}.
See the following theorem.

\begin{theorem}
    Let $C$ be a cyclic code with parameters $[n,k,d]_{q^2}$. Then there is an QUENTA code
    with parameters $[[n,k-|Z(C^{\perp_h})\cap Z(C)|,d;n-k-|Z(C^{\perp_h})\cap Z(C)|]]_q$.
    \label{newConstruction:thm2}
\end{theorem}

\begin{proof}
            First of all, from Proposition~\ref{Intersection_Cyclic} we have
            $\dim(C^\perp\cap C) = n - |Z(C^\perp)\cup Z(C)| = n-|Z(C)|-|Z(C^{\perp_h})|+|Z(C^{\perp_h})\cap Z(C)| =
            k - k + |Z(C^{\perp_h})\cap Z(C)| = |Z(C^{\perp_h})\cap Z(C)|$. So,
            $c = \dim(C^{\perp_h}) - \dim(C^\perp\cap C) = n-k - |Z(C^{\perp_h})\cap Z(C)|$. Using a $[n,k,d]_{q^2}$ to
            construct a QUENTA codes via Proposition~\ref{Prep:WildeHerm}, we derive a code with parameters
            $[[n,k-|Z(C^{\perp_h})\cap Z(C)|,d;n-k-|Z(C^{\perp_h})\cap Z(C)|]]_q$.
\end{proof}

\begin{corollary}
    Let $C$ be an LCD cyclic code with parameters $[n,k,d]_{q^2}$. Then there is a
    maximal entanglement QUENTA code with parameters $[[n,k,d;n-k]]_q$.
    \label{newConstruction:corollary2}
\end{corollary}

\begin{proof}
            From the proof of Theorem~\ref{newConstruction:thm2}, we have that
            $\dim(C^{\perp_h}\cap C) = |Z(C^{\perp_h})\cap Z(C)|$. Since that $C$ is LCD,
            $|Z(C^{\perp_h})\cap Z(C)| = 0$ and the result follows from Theorem~\ref{newConstruction:thm2}.
\end{proof}

Differently of constructing QUENTA code from Euclidean dual cyclic code, the construction via Hermitian dual can be more
delicate, even for Reed-Solomon codes. Even so, we are going to show that is possible to construct QUENTA codes from a particular
case of Reed-Solomon codes and some cyclic codes.

%
%

\begin{theorem}
      Let $q$ be a prime power and assume $C = RS_{k}(n,1)$ be a Reed-Solomon codes over $\mathbb{F}_{q^2}$
      with $k = qt + r < q^2$, where $t\geq 1$ and $0 \leq r \leq q-1$, and $n = q^2$.
      Then we have the following:
      \begin{itemize}
	  \item If $t\geq q-r-1$, then there exists an MDS QUENTA code with parameters
	  $$[[q^2, (t+1)^2 - 2(q-r) + 1, q(q-t)-r+1; (q-t-1)^2+1]]_q.$$
	  \item If $t< q-r-1$, then there exists an MDS QUENTA code with parameters
	  $$[[q^2, t^2 - 1, q(q-t)-r+1; (q-t)^2 - 2r-1)]]_q.$$
      \end{itemize}
      \label{Theorem:RS_Hermit}
\end{theorem}

\begin{proof}
      Since that $C = RS_{k}(n,0)$, we have that $Z(C) = \{0, 1, 2, \ldots, n - k - 1\}$. From the proof of
      Theorem~\ref{Cyclic_Hermitian_Dual}, we also have that $Z(C^{\perp_h}) = qZ(C^\perp) = \{q, 2q, \ldots, kq\}$.
      From $n = q^2$ and $k = qt + r$, we can rewrite these two defining set as
      $Z(C) = \{qi+j|0\leq i \leq q - t - 2, 0\leq j \leq q-1\}\cup\{(q-t-1)q + j| 0\leq j \leq q-r-2\}$ and
      $Z(C^{\perp_h}) = \{qi+j|0\leq i \leq q-1, 0\leq j\leq t-1\}\cup \{qi+t|0\leq i \leq r\}$. Using this description,
      we can compute $|Z(C)\cap Z(C^{\perp_h})|$. To do so, we have to consider
      two cases separately, $t\geq q-r-1$ and $t< q-r-1$. For the first case, the intersection is given by the following set
      $Z(C)\cap Z(C^{\perp_h}) = \{qi+j| 0\leq i \leq q-t-2, 0\leq j \leq t\}\cup \{(q-t-1)q+j|0\leq j \leq q-r-2\}$. Thus,
      $|Z(C)\cap Z(C^{\perp_h})| = (q-t-1)(t+1) + q-r-1$. Similarly for the case $t< q-r-1$, we have
      $Z(C)\cap Z(C^{\perp_h}) = \{qi+j| 0\leq i \leq q-t-1, 0\leq j \leq t-1\}\cup \{qi+t| 0\leq i \leq r\}$, which implies in
      $|Z(C)\cap Z(C^{\perp_h})| = (q-t)t + r+1$. Applying the previous computations, and using the fact that
      $C$ has parameters $[q^2, k, q^2 - k + 1]_{q^2}$, to Theorem~\ref{newConstruction:thm2}, we have that there exists a QUENTA code with parameters

      \begin{itemize}
          \item $[[q^2, (t+1)^2 - 2(q-r) + 1, q(q-t)-r+1; (q-t-1)^2+1]]_q$, for $t\geq q-r-1$; and
          \item $[[q^2, t^2 - 1, q(q-t)-r+1; (q-t)^2 - 2r-1)]]_q$, for $t< q-r-1$.
      \end{itemize}
\end{proof}

\begin{theorem}
        Let $n = q^4 - 1$ and $q\geq 3$ a prime power. There exists an QUENTA code with parameters
        $[[n, n - 4(a-1)-3, d\geq a+1;1]]_q$, where $2\leq a\leq q^2 - 1$.
        \label{newConstruction:HermitianBCH}
\end{theorem}

\begin{proof}
        Let $C_a$ be a cyclic code with defining set $Z(C_a) = \mathbb{C}_0\cup\mathbb{C}_{q^2 + 1}\cup(\cup_{i=2}^a\mathbb{C}_{q^2 + a})$, for
        $2\leq a \leq q^2-1$. From Ref. \cite{Giuliano:2009}, we have that $\mathbb{C}_{q^2 + 1} = \{q^2 + 1\}$ and
        $\mathbb{C}_{q^2 + a} = \{q^2 + a, 1+aq^2\}$. It is trivial to show that $\mathbb{C}_{0} = \{0\}$.
        From $-qZ(C_a)\cap Z(C_a) = \mathbb{C}_0$ \cite{Giuliano:2009}, we can see that
        $Z(C_a^{\perp_h})\cap Z(C_a) = Z(C_a)\setminus\mathbb{C}_0$. Hence,
        $|Z(C_a^{\perp_h})\cap Z(C_a)| = 2(a-1)+1$. From the assumption of the defining set, the dimension
        and minimal distance of the classical code is $k = n - 2(a-1)-2$ and $d\geq a+1$, respectively.
        Thus, applying these quantities to Theorem~\ref{newConstruction:thm2}, we have that there exists an QUENTA code
        with parameters $[[n, n - 4(a-1)-3, d\geq a+1;1]]_q$.
\end{proof}

Two important statements can be said about Theorem~\ref{newConstruction:HermitianBCH}. Comparing the bound given for
the minimal distance and the Singleton bound for QUENTA codes,
we see that the difference between these two values is equal to $a-1$. So, for lower values of $a$ (such as $a = 2$ or $a = 3$)
the QUENTA codes have minimal distance close to optimal; e.g., if $a = 2$ (or $a = 3$), the family of QUENTA codes is almost MDS
(or almost near MDS). The second point is that the codes in Theorem~\ref{newConstruction:HermitianBCH} can be seen as
a generalization of the result by Qian and Zhang \cite{Qian:2019}.

In the following, we use LCD cyclic code to construct maximal entanglement QUENTA codes. The families obtained have
an interesting range of possible parameters.

\begin{theorem}
            Let $q$ be a prime power, $m\geq 2$, $2\leq \delta \leq q^{2\lceil\frac{m}{2}\rceil}+1$, and $\kappa = q^{2m}-2-2(\delta - 1 - \Big{\lfloor}\frac{\delta-1}{q^2}\Big{\rfloor})m$.
            Then,
            \begin{enumerate}
                \item For $m$ odd and $1\leq u\leq q-1$, there is a maxima entanglement QUENTA code with parameters
                $[[q^{2m}-1, k, d\geq \delta + 1 + \lfloor\frac{\delta-1}{q}\rfloor; q^{2m}-1-k]]_q$, where
                \begin{eqnarray}
                            k =
                            \left\{
                              \begin{array}{ll}
                                \kappa,               & \text{if } 2\leq \delta \leq q^m -1; \\
                                \kappa + u^2 m,       & \text{if } uq^m\leq \delta \leq (u+1)(q^m - 1); \\
                                \kappa + (u^2+2v+1)m, & \text{if } \delta = (u+1)(q^m - 1)+ v + 1\text{ for }0\leq v\leq u-1; \\
                                \kappa + q^2 m,       & \text{if } \delta = q^{m+1}\text{ or }q^{m+1}+1.
                              \end{array}
                            \right.
                            \label{Eq:CyclicHermitianEq1}
                \end{eqnarray}
                \item For $m$ even, there is an maximal entanglement QUENTA code with parameters
		\begin{equation}
			[[q^{2m}-1, \kappa, d\geq \delta + 1 + \lfloor\frac{\delta-1}{q}\rfloor;
			2(\delta - 1 - \lfloor\frac{\delta-1}{q^2}\rfloor)m+1]]_q.
			\label{Eq:CyclicHermitianEq2}
		\end{equation}
            \end{enumerate}
	    \label{newConstruction:maximalEntanglementCyclicCodes}
\end{theorem}

\begin{proof}
            From Li~\cite{Li:2018}, we have that there are LCD cyclic codes with parameters
            $[q^{2m}-1, k, \delta + 1 + \lfloor\frac{\delta-1}{q}\rfloor]_{q^2}$, where $k$ is the same as in
            Eqs.~\ref{Eq:CyclicHermitianEq1}~and~\ref{Eq:CyclicHermitianEq2} for $m$ odd and even, respectively. Thus,
            applying this LCD code to Corollary~\ref{newConstruction:corollary2} we obtain the mentioned codes.
\end{proof}

%
%

\section{Code Examples}
\label{sec:codComp}
In Tables~\ref{table2}, we present some MDS QUENTA codes obtained
from Corollary~\ref{MDS_QUENTA_RS} and Theorem~\ref{Theorem:RS_Hermit}. The codes in the first column are obtained from
the Euclidean construction and the ones in the second from the Hermitian construction. As can be seen, the later
one has a higher length when compared with the same field. So, they can be used in applications where the quantum system
has low degrees of freedom. On the other hand, the codes in the first column can reach values that the ones from the
Hermitian construction cannot. Thus, these two class of QUENTA codes are suitable for their specific applications.

\begin{table}[h]
\begin{center}
\caption{Some new MDS maximal entanglement QUENTA codes from Reed-Solomon codes\label{table2}}
\begin{tabular}{c|c}
\hline\hline
New QUENTA codes -- Corollary~\ref{MDS_QUENTA_RS} & New QUENTA codes -- Theorem~\ref{Theorem:RS_Hermit}\\
$[[n,2b - 1, n-k+1; n+2b-2k-1]]_q$                & $[[q^2, t^2 - 1, q(q-t)-r+1; (q-t)^2 - 2r-1)]]_q$  \\
$0< b\leq (k+1)/2$ and $0<k<n\leq q$           & $qt + r < q^2$, where $1\leq t< q-r-1$ and $0\leq r \leq q-1$\\
\hline \multicolumn{2}{c}{\hspace{-2.28cm}Examples}\\ \hline
\hline ${[[3, 1, 3; 2]]}_{3}$                     & ${[[16, 3, 9; 3]]}_{4}$ \\
\hline ${[[4, 3, 2; 1]]}_{4}$                     & ${[[64, 35, 17; 3]]}_{8}$ \\
\hline ${[[7, 3, 5; 4]]}_{7}$                     & ${[[64, 15, 31; 11]]}_{8}$ \\
\hline ${[[8, 5, 4; 3]]}_{8}$                     & ${[[256, 196, 33; 3]]}_{16}$ \\
\hline ${[[11, 9, 3; 2]]}_{11}$                   & ${[[256, 120, 78; 18]]}_{16}$ \\
\hline ${[[13, 9, 5; 4]]}_{13}$                   & ${[[1024, 784, 129; 15]]}_{32}$ \\
\hline ${[[16, 13, 3; 2]]}_{16}$                  & ${[[1024, 624, 220; 38]]}_{32}$ \\
\hline
\end{tabular}
\end{center}
\end{table}

One family of QUENTA codes derived from BCH codes have been constructed, see Theorem~\ref{newConstruction:EuclideanBCH}.
Some examples of these QUENTA codes are shown in Table~\ref{table3}. As can be seen from Table 1 in
Ref.~\cite{Luo:2019} (and the reference there in), the QUENTA codes derived from
Theorem~\ref{newConstruction:EuclideanBCH} have new parameters when compared with QUENTA codes in the literature.
So, even not having good parameters as the ones in our Table~\ref{table2}, these codes are new.
One advantage of our codes from the ones in the literature is that, since they were constructed from two
BCH codes not only one as is commonly found in the literature, we have more liberty in the choice of parameters.
Such property can help to adjust the QUENTA code to the framework where it will be used.
Please, see Table~\ref{table3} for some examples.

\begin{table}[h]
\begin{center}
\caption{Some new QUENTA codes from BCH codes\label{table3}}
\begin{tabular}{c}
\hline\hline
New QUENTA codes -- Theorem~\ref{newConstruction:EuclideanBCH}\\
$[[q^2 - 1, 2a+1, b+1; 2b-\lfloor \frac{b}{q}\rfloor]]_q$\\
$1\leq b\leq q$ and $0\leq a< q-b$\\
\hline Examples\\ \hline
\hline ${[[15, 5, 2; 2]]}_{4}$\\
\hline ${[[48, 9, 3; 4]]}_{7}$\\
\hline ${[[63, 7, 5; 8]]}_{8}$\\
\hline ${[[80, 13, 3; 4]]}_{9}$\\
\hline ${[[255, 19, 7; 12]]}_{16}$\\
\hline
\end{tabular}
\end{center}
\end{table}

The remaining QUENTA codes constructed in this paper are the ones derived from cyclic codes that are neither
Reed-Solomon nor BCH codes. Two families of such codes were created. Both of them using Hermitian construction.
Some examples of parameters that can be obtained from these codes are presented in Table~\ref{table4}. Codes in the
first column are almost MDS or almost near MDS; i.e., the Singleton defect, which is the difference between the quantum
Singleton bound (QSB) presented in Eq.~\ref{QSB} and the minimal distance of the code, is equal to one or two units.
Lastly, it is displayed in the second column of Table~\ref{table4} some codes from
Theorem~\ref{newConstruction:maximalEntanglementCyclicCodes}. All codes in
Theorem~\ref{newConstruction:maximalEntanglementCyclicCodes} are maximal entanglement. Thus, they can be used to
achieve the hashing bound \cite{Li:2014}. Having length proportional to a high power of the cardinality of the field,
it is expected to achieve low error probability using these codes. Comparing our parameters with the ones in
the literature \cite{Lu:2018,Guo:2015,Lu:2015,Lv:2015}, we see that our codes are new.

\begin{table}[h]
\begin{center}
\caption{Some QUENTA codes from Cyclic Codes via Hermitian Construction\label{table4}}
\begin{tabular}{c|c}
\hline\hline
New QUENTA codes -- Theorem~\ref{newConstruction:HermitianBCH}                 & New QUENTA codes -- Theorem~\ref{newConstruction:maximalEntanglementCyclicCodes}    \\
\hline \multicolumn{2}{c}{Examples}\\ \hline
\hline ${[[80, 73, 3; 1]]}_{3}$                                                & ${[[80, 42, 14; 38]]}_{3}$ \\
\hline ${[[80, 69, 4; 1]]}_{3}$                                                & ${[[80, 50, 10; 30]]}_{3}$ \\
\hline ${[[255, 248, 3; 1]]}_{4}$                                              & ${[[255, 193, 20; 62]]}_{4}$ \\
\hline ${[[255, 244, 4; 1]]}_{4}$                                              & ${[[255, 237, 7; 18]]}_{4}$ \\
\hline
\end{tabular}
\end{center}
\end{table}

\section{Conclusion}
\label{sec:Conclusion}
This paper has been devoted to the use of cyclic codes in the construction of QUENTA codes.
General construction methods of QUENTA codes from cyclic codes via defining sets
have been presented, using both Euclidean and Hermitian dual of the classical codes.
As an application of these methods, five families of QUENTA codes were created. Two
of them were derived from Reed-Solomon codes, which resulted in MDS codes. An additional
family of almost MDS or near almost MDS QUENTA codes was derived from general cyclic codes.
One of the remaining family used BCH codes as the classical counterpart. The construction of
this family of QUENTA code used two BCH codes, which provided more liberty in the parameters of
the quantum code. Lastly, we construct a family of maximal entanglement QUENTA code that can be
useful in reaching the hashing bound.


\section{Acknowledgements}
I am grateful to Ruud Pellikaan for proposing this research problem and for
interesting discussions which helped me to clarify some points of view.
This work was supported by the \emph{Conselho Nacional de Desenvolvimento
Cient\'ifico e Tecnol\'ogico}, grant No. 201223/2018-0.

\bibliographystyle{splncs04}
\bibliography{ref}

\end{document}